\documentclass[12pt]{article}
\usepackage{amsmath}
\usepackage{amssymb}
\usepackage{amsthm}
\usepackage{fullpage}
\usepackage[utf8]{inputenc}
\usepackage{mathtools}
\usepackage{url}
\usepackage[usenames,svgnames]{xcolor}
\usepackage{cite}

\usepackage[pagebackref=false]{hyperref} 
\usepackage[all]{hypcap} 

\theoremstyle{plain}
\newtheorem{theorem}{Theorem}[section]
\newtheorem{lemma}[theorem]{Lemma}

\newtheorem{example}{Example}[section]
\newtheorem{examples}{Example}[subsection]

\newtheorem{remark}{Remark}[section]

\theoremstyle{definition}
\newtheorem{definition}{Definition}[section]

\numberwithin{equation}{section} 

\hypersetup{
breaklinks,
colorlinks,
citecolor=Green,
linkcolor=BlueViolet,
urlcolor=DarkBlue,
pdftitle={Quantum Hurwitz numbers and MacDonald polynomials},
pdfauthor={John Harnad }, }

\DeclareMathOperator{\cyc}{cyc}

\newcommand{\id}{\mathrm{Id}}

\DeclareMathOperator{\aut}{aut}

\DeclarePairedDelimiter{\abs}{|}{|}

\def\be{\begin{equation}}
\def\ee{\end{equation}}

\def\bea{\begin{eqnarray}}
\def\eea{\end{eqnarray}}

\def\bt{\begin{theorem}}
\def\et{\end{theorem}}

\def\bex{\begin{example}\small \rm}
\def\eex{\end{example}}

\def\bexs{\begin{examples}\small \rm}
\def\eexs{\end{examples}}

\def\ra{\rightarrow}
\def\ss{\subset}
\def\deq{\coloneqq}

\def\br{\begin{remark}\small \rm}
\def\er{\end{remark}}

\def\&{&{\hskip -20pt}}

\def\JJ{\mathcal{J}}

\def\ZZ{\mathcal{Z}}

\def\Cb{\mathbf{C}}

\def\Nb{\mathbf{N}}
\def\Pb{\mathbf{P}}
\def\Rb{\mathbf{R}}

\def\Zb{\mathbf{Z}}

\begin{document}
\baselineskip 16pt
\medskip
\begin{center}
\begin{Large}\fontfamily{cmss}
\fontsize{17pt}{27pt}
\selectfont
\textbf{Quantum Hurwitz numbers and Macdonald polynomials}\footnote{Work of J.H. supported by the Natural Sciences and Engineering Research Council of Canada (NSERC) and the Fonds Qu\'ebecois de la recherche sur la nature et les technologies (FQRNT). }
\end{Large}\\
\bigskip
\begin{large}  {J. Harnad}
 \end{large}
\\
\bigskip
\begin{small}
{\em Centre de recherches math\'ematiques,
Universit\'e de Montr\'eal\\ C.~P.~6128, succ. centre ville, Montr\'eal,
Qu\'ebec, Canada H3C 3J7 }\\
\smallskip
{\em Department of Mathematics and
Statistics, Concordia University\\ 1455 de Maisonneuve Blvd. W. 
Montreal, Quebec,  Canada H3G 1M8 } 
\end{small}
\end{center}
\bigskip

\begin{abstract}
Parametric families in   the center $\Zb(\Cb[S_n])$  of the group algebra of the symmetric group are 
obtained by identifying  the indeterminates in the generating function for Macdonald polynomials  as  
commuting Jucys-Murphy elements.  Their eigenvalues  provide coefficients in the double Schur function 
expansion  of  2D Toda  $\tau$-functions of  hypergeometric type.   Expressing these in the basis of  products of power sum symmetric functions,  the coefficients may be interpreted geometrically  as  parametric families
of quantum Hurwitz numbers, enumerating weighted branched coverings of the Riemann sphere. Combinatorially, they give quantum weighted sums over paths in the Cayley graph of $S_n$  generated by transpositions.   Dual pairs of bases for the algebra of   symmetric functions with  respect to the scalar product in which the Macdonald polynomials  are orthogonal  provide both  the geometrical and combinatorial significance of  these 
quantum  weighted enumerative invariants.
   \end{abstract}

\tableofcontents

\section{Introduction: weighted Hurwitz numbers }

A new method for constructing parametric families of 2D Toda $\tau$-functions \cite{Ta, Takeb, UTa} of hypergeometric type \cite{OrS}   that serve as generating functions for various  types of weighted Hurwitz numbers was developed  in \cite{GH1, GH2, H1,  H2, HO2}.  This was originally inspired by the work of Pandharipande \cite{P} and Okounkov \cite{Ok}, which first used a special case of KP and 2D-Toda $\tau$-functions as  generating functions for single and double  for Hurwitz numbers  when all branchings other than the ones specified at one or two points rare equired to be simple, and the weighting for these is uniform. The general case  gives infinite parametric families of weighted  enumerations of $n$-fold  branched coverings  of the  Riemann sphere or, equivalently, weighted  paths in  the Cayley graph  of the symmetric group $S_n$  generated by transpositions. They are  derived  from  parametric families of weight generating functions by defining associated symmetric  functions of an arbitrary number of indeterminates multiplicatively.  Replacing one set of  indeterminates in  the Cauchy-Littlewood  generating function \cite{Mac}  by the commuting elements of the group algebra introduced   by Jucys \cite{Ju} and Murphy \cite{Mu}, while evaluating  the  other set at parameter values  defining the weightings provides parametric families of elements of the center  $\Zb(\Cb[S_n])$ of the  group algebra. Expanding these as sums over  products of dual bases of the algebra of symmetric functions,  and applying them multipicatively  to the basis of  the center of the group algebra consisting of cycle-type sums $\{C_\mu\}$ leads to an identification  of both the geometrical significance of the weighted  Hurwitz numbers and the combinatorial  one.  
 
It was shown in \cite{GH1, GH2, HO2}  that all previously studied examples  of generating functions 
for Hurwitz numbers    \cite{P, Ok, GGN1, GGN2, BEMS, Z, KZ, AMMN, AC1,  AC2}  may be viewed as special cases of 
this general construction, and several new examples were introduced,  including three forms of quantum 
Hurwitz number  \cite{GH2} and their multispecies generalization \cite{H1}. Other notions of weighted or 
quantum Hurwitz numbers  have also been considered, including those for branched 
coverings of $\Rb \Pb^2$, whose generating functions are BKP $\tau$-functions \cite{NOr1, NOr2}, 
and Hurwitz numbers enumerating factorization of Singer cycles \cite{LRS}. 

  In the following, we extend the special class of  weighted Hurwitz numbers 
 introduced in \cite{GH1, GH2} by introducing an additional pair $(q,t)$ of deformation parameters
 in the definition of the weight generating functions. The result is to replace the Cauchy-Littlewood  formula,
 which generates  dual bases in the algebra of the symmetric functions with respect to the standard scalar product pairing by the corresponding one  for MacDonald   polynomials \cite{Mac}.   In \autoref{hypergeometric_tau} the general method is developed,  and used to derive an infinite  parametric family of 2D Toda $\tau$-functions of hypergeometric type depending not only on the previously introduced classical weight determining parameters, but also the additional pair $(q,t)$ of quantum deformation parameters entering in the definition of the scalar product. These are shown to be generating functions  for infinite parametric  families of quantum weighted Hurwitz numbers when expanded in the basis of  products of   power sum symmetric functions. The combinatorial significance, in terms of quantum weighted paths in the Cayley graph, is derived in \autoref{combinatorial_hurwitz}, and the  geometric one,  in terms of quantum weighted enumeration of branched  covers, in \autoref{geometric_hurwitz}.   \autoref{examples} is devoted to various examples obtained by  specialization of the parameters and taking limits. It is shown  how all previously studied cases of weighted   Hurwitz numbers, whether classical or quantum, may be recovered within this more general setting,   and a number of new families  are added,  including those associated  to Hall-Littlewood  and Jack polynomials.
 
\section{2D Toda $\tau$-functions,  MacDonald polynomials and quantum weighted Hurwitz numbers}
\label{hypergeometric_tau}

\subsection{The generating function for Macdonald polynomials}

Following \cite{Mac} (sec. VI. 2), we define a $2$-parameter family of scalar products $(\ , \ )_{(q, t)}$ on the algebra $\Lambda$ 
of symmetric functions in an infinite number of indeterminates ${\bf x} = (x_1, x_2, \dots  )$, such that
the power sum symmetric functions are orthogonal; 
\be
( p_\lambda, p_\mu)_{(q, t)} := z_\mu(q,t) \delta_{\mu, \nu}
\ee
where 
\be
p_\lambda := \prod_{i=1}^{\ell(\lambda} p_{\lambda_i}  \in  \Lambda, \qquad
p_j  := \sum_{i}x_i^j, \quad j\in \Nb
\ee
are the power sum symmetric functions corresponding to the integer partition
\be
\lambda= \lambda_1 \ge \cdots \ge \lambda_{\ell(\lambda)}
\ee
 of length $\ell(\lambda)$. The normalization factor $z_\mu(q,t)$ is defined as
 \be
z_\mu(q,t) := z_\mu n_\mu (q, t),  \quad 
z_\mu :=  \prod_{i=1}^{\mu} i^{m_i(\mu)}(m_i(\mu))! .
\ee
where $m_i(\mu)$ is the number of parts of $\mu$ equal to $i$
and
\be
n_\mu(q, t) :=  \prod_{i=1}^{\ell(\mu)} {1 - q^{\mu_i} \over 1- t^{\mu_i}}.
\ee

 The Macdonald polynomials  $\{P_\lambda({\bf x}, q, t)\}$ may be defined  \cite[Chapt. VI]{Mac} as
  the unique basis for $\Lambda$ determined  by  two conditions: orthogonality with respect to
  the scalar product $( \, , \, )_{(q, t)}$
  \be
  (P_\lambda, P_\mu)_{(q, t)} = 0 \quad \text{  if } \lambda \neq \mu,
  \ee
  and lower  triangular normalized decomposition  (with respect to the dominance partial ordering \cite[Sec. I.1, pg. 7]{Mac})
 in the basis $\{ m_\lambda\}$ of monomial symmetric functions
  \be
  P_\lambda({\bf x}, q,t ) = m_\lambda + \sum_{\mu < \lambda} K_{\lambda \mu} (q,t) m_\mu ({\bf x}).
  \ee
  
 The generating function \cite{Mac}
  \be
  \Pi({\bf x}, {\bf y}, q, t) := \prod_{ij} {(t x_i y_j; q)_\infty \over (x_i y_j; q)_\infty}
  \label{Pi_xy_qt}
  \ee
    where
  \be
  (u; q)_\infty := \prod_{k=0}^\infty (1 - u q^k)
  \ee
  is the (infinite) quantum Pochhammer symbol, has the following alternative expansions \cite[Sec.~VI.2]{Mac}
    in terms of products of symmetric functions in the pair of infinite sequences
  of determinate ${\bf x}= (x_1, x_2, \dots )$, ${\bf y } = (y_1, y_2, \dots )$,
  \bea
  \Pi({\bf x}, {\bf y}, q, t)&\& = \sum_\lambda b_\lambda(q,t) P_\lambda ({\bf x}, q,t)  P_{\lambda} ({\bf y}, q,t)  \\
  &\& = \sum_\lambda  g_\lambda ({\bf x}, q,t) \,  m_\lambda ({\bf y})  \\
  &\& = \sum_\lambda m_\lambda ({\bf x}) \,  g_\lambda ({\bf y}, q,t) ,  
  \eea
  where
  \be
  b_\lambda(q,t) := (P_\lambda, P_\lambda)_{(q, t)}^{-1}
  \ee
    and
  \be
g_\lambda ({\bf x}, q,t) := \prod_{i=1}^{\ell(\lambda)} g_{\lambda_i} ({\bf x}, q,t),
\ee
where
\be
 g_j({\bf x}, q,t):= b_{(j)}(q,t) P_{(j)} ({\bf x}, q,t) = \sum_{\mu, \, |\mu|=j} z_\mu(q,t)^{-1} p_\mu({\bf x}).
  \ee
The basis $\{ g_\lambda({\bf x}, q, t)$ provides  the $(q,t)$ analog of the elementary $\{e_\lambda\}$ and complete 
$\{h_\lambda\}$ symmetric  function basis \cite{Mac}, interpolating between them  in the case of Hall polynomials ($q=0$).

\subsection{Quantum weight generating function}

We now proceed as in \cite{GH1, GH2}  to define parametric families within the center  $\Zb(\Cb[S_n])$ 
of the group algebra  $\Cb[S_n]$ by identifying the indeterminates $(x_1, x_2, \dots )$ with
a given set of constants $(c_1, c_2, \dots)$ and the $(y_1, y_2, \dots )$ with $z$ times
the  commuting Jucys-Murphy elements $\JJ:= (\JJ_1, \dots, \JJ_n)$ of $\Cb[S_n]$ \cite{Ju, Mu, DG}, defined as  :
\be
\JJ_b\deq \sum_{a=1}^{b-1}(a\,b), \quad b=1, \dots, n, \quad n \in \Nb^+.
\ee

We define the {\em quantum weight generating function}  as
\be
M(q,t, {\bf c}, z ) := \ \prod_{i=1}^\infty M(q,t, z c_i) =  \sum_{j=0}^\infty g_j({\bf c}, q,t) z^j.
\label{M_qtcz}
\ee
where
\be
M(q, t,  z):= {(tz; q)_{\infty} \over (z; q)_{\infty}} =  \prod_{k=0}^\infty {1 - t  z q^k  \over 1- z q^k }.
\label{M_qtz}
\ee

The Jucys-Murphy elements generate a commuting subalgebra of the group
algebra $\Cb[S_n]$, and any symmetric polynomial in them is in the center  $\Zb(\Cb[S_n])$,
The resulting central element $ M_n(q,t, {\bf c}, z\JJ)  \in  \Zb(\Cb[S_n])$ is 
\bea
 M_n(q,t, {\bf c}, z\JJ) &\&:=  \prod_{a=1}^n M(q,t, {\bf c}, z\JJ_a) = \Pi({\bf c}, z\JJ, q, t)  
\label{Pihat}
\\
&\& = \sum_\lambda z^{|\lambda |} g_\lambda ({\bf c}, q,t)  m_\lambda (\JJ) 
\label{Pihat_combinatorial} \\
&\& = \sum_\lambda z^{|\lambda |}m_\lambda ({\bf c})  g _\lambda (\JJ, q,t).
\label{Pihat_geometrical} 
\eea

  
\subsection{Bases for $\Zb(\Cb[S_n])$ and the eigenvalues of $M_n(q, t, {\bf c},   z\JJ)$}
Proceeding as in \cite{GH1, GH2, HO2}, we make use of two standard bases of $\Zb(\Cb[S_n])$,
both labelled by partitions of $n$. The first consists of the cycle-type sums $\{C_\mu\}$:
\be
C_\mu = \sum_{h\in \cyc(\mu)} h,
\ee
where $\cyc(\mu) \ss S_n$ denotes the conjugacy class consisting of elements whose cycle
lengths are equal to the parts $\mu_i$ of the partition $\mu$. 

The second consists of the orthogonal idempotents $\{F_\lambda\}$,  corresponding to the irreducible representations of $S_n$,  labelled by partitions $\lambda$ of weight $|\lambda|=n$. These are linearly related to the cycle-type sums through  the equivalent, under the characteristic map, of the Frobenius character formula \cite{Frob, FH} 
\be
F_\lambda = h_\lambda^{-1}  \sum_{\mu, \, |\mu|=|\lambda|} \chi_\lambda(\mu) C_\mu, 
\quad  C_\mu =  z_\mu^{-1}\sum_{\lambda, \, \lambda| = |\mu|} \chi_\lambda(\mu)h_\lambda F_\lambda.
\label{F_lambda_C_mu}
\ee
Here $\chi_\lambda(\mu)$ is the  character of the irreducible representation of Young symmetry type $\lambda$,
evaluated on the class of cycle type $\mu$,  and 
\be
h_\lambda  := \det\left( {1 \over (\lambda_i - i +j)!}\right)^{-1}
\ee
is the product of the hook lengths of the Young diagram corresponding to the partition $\lambda$.

The elements  $F_\lambda$  satisfy the orthogonality relations
\be
F_\lambda F_\mu = F_\lambda  \ \delta_{\lambda \mu},
\label{F_lambda_orthog}
\ee
which imply that all elements of $\Zb(\Cb[S_n])$ act diagonally under multiplication in this base.
Eqs.~(\ref{F_lambda_C_mu}) and  (\ref{F_lambda_orthog}) imply
\be
C_\mu F_\lambda = {h_\lambda\chi_\lambda(\mu) \over z_\mu} F_\lambda,
\label{C_mu_F_lambda}
\ee
which means that the eigenvalue of $C_\mu$ on the basis element $F_\lambda$ is
the {\it central character}
\be
\phi_\lambda(\mu):={h_\lambda\chi_\lambda(\mu) \over z_\mu} .
\label{central_char}
\ee

It is a basic property \cite{Ju, Mu, DG}  that the eigenvalues of any central element
 $G(\JJ)  \in \Zb(\Cb[S_n])$  formed from a symmetric function
 $G\in \Lambda$ by identifying the indeterminates with the Jucys-Mulrphy elements
 are given by  evaluating $G$ on the content $\{j-i\}_{(i,j) \in \lambda}$
 of the partition $\lambda$
 \be
 G(\JJ)F_\lambda = G(\{j-i\}) F_\lambda, \quad (i,j) \in \lambda. 
 \ee
In particular, if $G({\bf x})$ is formed from a product of the same expression in each of the
 indeterminates ${\bf x} = (x_1, x_2, \dots)$
 \be
 G({\bf x}) = \prod_{i} g(x_i),
 \ee
 the eigenvalues of $G(\JJ)$ are given by the content product formula 
 \be
G(\JJ)F_\lambda = \prod_{(i,j) \in \lambda} g(j-i)F_\lambda.
 \label{content_product_eigenvalue}
 \ee
 
It follows that the eigenvalues $r_\lambda^{M(q, t, {\bf c}, z)}$ of  $M_n( q, t, {\bf c}, z\JJ) $ are
 given by the content product formula obtained from evaluation of the  generating function $M(q,t,{\bf c}, z)$
 at $\{z(j-i)\}$
 \bea
M_n(q,t, {\bf c}, z\JJ)  F_\lambda &\&= r_\lambda^{M(q,t, {\bf c}, z)}F_\lambda, 
 \label{Pi_hat_r_lambda}\\
r_\lambda^{M(q, t, {\bf c}, z)}  &\&=\prod_{(i, j)\in \lambda} M(q,t, {\bf c}, z(j-i)) 
=\prod_{(i,j)\in \lambda}\prod_{k=1}^\infty {(tz(j-i) c_k; q)_\infty \over (z(j-i)c_k; q)_\infty }.
\label{r_lambda_q_p}
\eea

More generally, for an arbitrary integer $N \in \Zb$, we define
\be
r_\lambda^{M(q,t, {\bf c}, z)}(N) := r_0^{M(q,t, {\bf c}, z)}(N)\prod_{(i,j)\in \lambda} M(q,t, {\bf c}, z(N+j-i)),
\label{r_lambda_G} 
 \ee
where
\bea
r^{M(q,t, {\bf c}, z)}_0(N)  &\&:= \prod_{j=1}^{N-1} M(q,t, {\bf c}, (N-j)z),  \quad r_0^{M(q,t, {\bf c}, z)}(0)  := 1, 
\cr
 r^{M(q,t, {\bf c}, z)}_0(-N) &\&:= \prod_{j=1}^{N} M^{-1}(q,t, {\bf c}, (j-N)z),
\quad N\geq 1,
\label{content_product_G_N}
\eea
and hence 
\be
r_\lambda^{M(q,t, {\bf c}, z)} = r_\lambda^{M(q,t, {\bf c}, z)} (0) .
\ee

\subsection{The 2D Toda $\tau$-function $\tau^{M(q,t,  {\bf c}, z)}(N, {\bf t}, {\bf s})$ as 
generating function for double quantum Hurwitz numbers $F^d_{M(q,t, {\bf c})} (\mu, \nu)$,  $H^d_{M(q,t, {\bf c})} (\mu, \nu)$}
\label{double_quantum_hurwitz}

The general theory  \cite{Ta, Takeb, UTa, OrS, HO1} implies that the following diagonal
double Schur function  expansion 
\be
\tau^{M(q,t, {\bf c}, z)}(N, {\bf t}, {\bf s}) := \sum_{\lambda} r_\lambda^{M(q,t, {\bf c}, z)}(N) s_\lambda({\bf t}) s_\lambda({\bf s}),
\label{tau_G_double_schur}
\ee
defines a 2D Toda $\tau$-function of  hypergeometric type, where
\be
{\bf t} = (t_1, t_2, \dots), \quad {\bf s} = (s_1, s_2, \dots)
\ee
are the 2D Toda flow variables, which may be identified in this notation in terms of the power sums 
\be
t_i = \frac{p_i}{i}, \quad s_i = \frac{p'_i}{i}
\ee
in two independent sets of variables.

We now apply the procedure developed in \cite{GH2} for deriving both the geometrical and 
combinatorial versions of weighted Hurwitz numbers associated to  a weight generating function $G(z)$.
Recall that the pure Hurwitz numbers $H(\mu^{1)}, \dots \mu^{(k)})$ may be viewed either
as the number of $n$-sheeted branched coverings of the Riemann sphere having $k$ branch points with
ramification profiles given by the partitions $\{\mu^{i)}\}_{i=1,\dots k}$ of lengths $|\mu^{(i)}| =n$, weighted
by the inverse of the order of the automorphism group or, equivalently, as the number of ways in which the
identity element in $S_n$ can be expressed as a product of elements belonging to the conjugacy
classes $\{\cyc(\mu^{i)}\}$.  A convenient way to express the latter is through the formula
\be
 H(\mu^{(1)}, \dots, \mu^{(k)}) = \frac{1}{n!} [\id] \prod_{i=1}^k C_{\mu^{(i)}},
\label{combin_hurwitz_id}
\ee
where $[\id] $ means taking the component of the identity element within the cycle sum 
basis $\{C_\mu\}$ of $\ZZ(\Cb[S_n])$ or, more generally,
\be
 \prod_{i=1}^k C_{\mu^{(i)}} = \sum_{\nu, \, |\nu| = |\mu^{(i)}|} H(\mu^{(1)}, \dots , \mu^{(k)} ,\nu)\,  z_\nu C_\nu, 
\label{combin_hurwitz}
\ee
which is equivalent to the Frobenius-Schur formula  (see \cite[Appendix~A]{LZ}) as shown in \cite[Sec. 5.2]{GH2} )
\be
H(\mu^{(1)}, \dots, \mu^{(k)}) =\sum_{\lambda} h_\lambda^{k-2} \prod_{i=1}^k \frac{\chi_\lambda(\mu^{(i)})}{z_{\mu^{(i)}}}, 
\label{Frob_Schur}
\ee

Following \cite{GH2, HO2}, we now consider  two notions of quantum weighted Hurwitz numbers 
associated to the generating function (\ref{M_qtcz}): combinatorial and geometrical. 


\subsubsection{Combinatorial quantum weighted Hurwitz numbers \cite{GH2}}

\begin{definition}{\em Signature of paths  \cite{GH2}.}
For every $d$-step path in the Cayley graph of $S_n$ generated by transpositions,  $(a,b)$, $a< b$, starting at the 
conjugacy class $\cyc(\nu)$ and ending at the class $\cyc(\mu)$, define its {\em signature}
$\lambda$ as the partition of weight $|\lambda|=d$ whose parts are
equal to the number of times a particular second element $b=1, \dots, n$ appears
amongst the sequence of transpositions $(a_1b_1) \cdots (a_d b_d)$ forming
the path from an element $h\in \cyc(\nu)$ to $(a_1 b_1) \cdots (a_d b_d)h \in \cyc(\mu)$.
\end{definition}

We recall the following Lemma from \cite[Lemma 2.3]{GH2}
\begin{lemma}
\label{generating_weighted_paths}
Multiplication by $m_\lambda(\JJ)$ defines an endomorphism of $\Zb(\Cb[S_n])$ which, expressed in the $\{C_\mu\}$ basis, is given by
\be
m_\lambda(\JJ) C_\mu = \sum_{\nu, \, \abs{\nu}=\abs{\mu}} m^\lambda_{\mu \nu} \frac{z_\nu}{\abs{\nu}!} C_\nu,
\ee
where 
\be
\tilde{m}^\lambda_{\mu \nu}: = {|\lambda|! \over \prod_{i=1}^{\ell(|\lambda|} \lambda_i !}m^\lambda_{\mu \nu}
\ee
 is the total number of $\abs{\lambda}$-step paths in the Cayley graph of $S_n$ from 
$\cyc(\nu)$ to $\cyc(\mu)$ with signature $\lambda$.
\end{lemma}
Combining this with (\ref{Pihat_combinatorial}) gives
\be
M_n(q,t, {\bf c}, z\JJ) \, C_\mu = \sum_{d=0}^\infty z^d\sum_{\nu, |\nu|=|\mu|=n} F^d_{M(q,t,{\bf c})} (\mu, \nu) z_\nu C_\nu,
\label{Pi_Cmu_F_Gd}
\ee
where
\be
F^d_{M(q,t,{\bf c})}(\mu, \nu) \deq {1 \over \abs{n}!} \sum_{\lambda, \ \abs{\lambda}
=d} g_\lambda({\bf c},q,t) m^\lambda_{\mu \nu}
\label{Fd_Mqt_def}
\ee
is the quantum weighted combinatorial Hurwitz number for such paths, and
\be
M(q,t,{\bf c}) := M(q,t,{\bf c}, z)\vert_{z=1}.
\ee
(Note that, whereas the infinite product (\ref{M_qtcz}) defining the generating function $M(q,t, {\bf c}, z)$ need not necessarily
represent a convergent power series in $z$, if $|q|<1$ and $|c_i|<1$ for all $i$, these are i, in fact, convergent for all values of $z$.)

Combining this with the results of the previous section leads to our first main result: the 2D Toda
$\tau$-function defined in (\ref{tau_G_double_schur})  for $N=0$ is the generating function
 for the  quantum weighted combinatorial  Hurwitz number  (\ref{Fd_Mqt_def}).
\begin{theorem}
\label{combinatorial_hurwitz}
Expanding $ \tau^{M(q,t, {\bf c}, z)}({\bf t}, {\bf s}) :=\tau^{M(q,t, {\bf c})}(0, {\bf t}, {\bf s})$ in the basis consisting of products of power sum symmetric functions, the coefficients are the combinatorial quantum Hurwitz numbers  (\ref{Fd_Mqt_def}).
\be
\tau^{M(q,t, {\bf c}, z)} ({\bf t}, {\bf s})
= \sum_{d=0}^\infty \sum_{\substack{\mu, \nu \\ \abs{\mu} =
 \abs{\nu}}} z^d F^d_{M(q,t,{\bf c})}(\mu, \nu) p_\mu({\bf t}) p_\nu({\bf s})
\label{tau_Mqt_F}
\ee
\end{theorem}
\begin{proof}
Combining (\ref{Pi_Cmu_F_Gd}) with (\ref{F_lambda_C_mu}), and using  eq.~(\ref{Pi_hat_r_lambda}), 
gives
\be
\sum_{d=0}^\infty z^d \sum_{\nu,  |\nu|=|\mu|} F^d_{M(q,t,{\bf c})}(\mu, \nu) \chi_\lambda(\nu) = { \chi_\lambda(\mu) \over z_\mu} r_\lambda^{M(q,t, {\bf c}, z)}(N) 
\label{r_lambda_chi_lambda}
\ee
Substituting the Frobenius character formula
\be
s_\lambda({\bf t})= \sum_{\mu, \, |\mu| =|\lambda|}z_\mu^{-1} \chi_\lambda(\mu) p_\mu({\bf t} ) ,
\quad  s_\lambda({\bf s} )= \sum_{\nu, \, |\mu| =|\lambda|}z_\nu^{-1} \chi_\lambda(\nu) p_\nu({\bf s} ) ,
\ee
and (\ref{r_lambda_chi_lambda}) into (\ref{tau_G_double_schur}), for $N=0$, and using the 
orthogonally of characters,  we obtain (\ref{tau_Mqt_F}). 
\end{proof}

  
\subsubsection{Enumerative geometrical   quantum weighted Hurwitz numbers}

We recall the two types of weighting factors  appearing in the definition of quantum Hurwitz numbers
in ref.~\cite{GH2}.
\bea
W_{E(q)} (\mu^{(1)}, \dots, \mu^{(k)}) &\& \deq
{1\over \abs{\aut(\lambda)}}\sum_{\sigma\in S_k} \sum_{0 \le i_1 < \cdots < i_k}^\infty q^{i_1 \ell^*(\mu^{(\sigma(1))})} \cdots q^{i_k \ell^*(\mu^{(\sigma(k))})} \cr
&\&={1\over  \abs{\aut(\lambda)}}\sum_{\sigma\in S_k} \frac{q^{(k-1) \ell^*(\mu^{(\sigma(1))})} \cdots q^{\ell^*(\mu^{(\sigma(k-1))})}}{
(1- q^{\ell^*(\mu^{(\sigma(1))})}) \cdots (1- q^{\ell^*(\mu^{(\sigma(1))})} \cdots q^{\ell^*(\mu^{(\sigma(k))})})}, \cr
&\&
\label{W_E_q}
\eea
\bea
W_{H(q)} (\nu^{(1)}, \dots, \nu^{(\tilde{k})}) &\& \deq
{(-1)^{\ell^*(\tilde{\lambda})}\over \abs{\aut(\tilde{\lambda})}} \sum_{\sigma\in S_{\tilde{k}}} \sum_{0 \le i_1 \le \cdots \le i_{\tilde{k}}}^\infty q^{i_1 \ell^*(\nu^{(\sigma(1))})} \cdots q^{i_k \ell^*(\nu^{(\sigma({\tilde{k}}))})} \cr
&\&= {(-1)^{\ell^*(\tilde{\lambda})}\over   \abs{\aut(\tilde{\lambda})}}\sum_{\sigma\in S_{\tilde{k}}} \frac{1}{
(1- q^{\ell^*(\nu^{(\sigma(1))})}) \cdots (1- q^{\ell^*(\nu^{(\sigma(1))})} \cdots q^{\ell^*(\nu^{(\sigma({\tilde{k}}))})})}, \cr
&\&
\label{W_H_q}
\eea
where $\lambda$ is the partition with parts $(\ell^*(\mu^{(1)}), \dots, \ell^*(\mu^{(k)}))$, $\tilde{\lambda}$ 
 the one with parts $(\ell^*(\nu^{(1)}), \dots, \ell^*(\nu^{({\tilde{k}})}))$, and $\abs{\aut(\lambda)}$ is the order of the automorphism 
 group of $\lambda$ :
\be
\abs{\aut(\lambda)}= \prod_{i=1}^{\ell(\lambda)} m_i(\lambda)!.
\ee 
Denote the product of these
\be
W_q (\mu^{(1)}, \dots, \mu^{(k)}; \nu^{(1)}, \dots, \nu^{({\tilde{k}})})
 := W_{E(q)} (\mu^{(1)}, \dots, \mu^{(k)})W_{H(q)} (\nu^{(1)}, \dots, \nu^{({\tilde{k}})}) .
\ee

Recall also the definition of the Pochhammer symbol $(u)_\lambda $ associated with 
a partition $\lambda$
\be
(u)_\lambda := \prod_{i=1}^{\ell(\mu)}\prod_{j=1}^{\lambda_i}(u+j-i) 
\ee
and the following Lemma (cf.~\cite{OrS}), which follows from the Frobenius character formula.
\begin{lemma}
\label{Poch_Frob}
The Pochhammer symbol may be expressed as
\be
(u)_\lambda =s_\lambda({\bf t}(u)) h_\lambda
= \left(1+ h_\lambda\sideset{}{'}
\sum_{\mu, \, \abs{\mu}=\abs{\lambda}}\frac{\chi_\lambda(\mu)}{z_\mu} u^{-\ell^*(\mu)}  \right)
\label{pochhammer_frobenius}
\ee
where
\be
{\bf t}(u) :=  (u, \frac{u}{2}, \frac{u}{3},  \dots ), 
\label{t_u}
\ee
and $\sum'_{\mu, |\mu |=|\lambda|)}$ denotes the sum over all partitions other than the cycle type of the identity element $(1)^{\abs{\lambda}}$.
\end{lemma}

It is useful to know how any given symmetric combination of the Jucys-Murphy elements may be 
represented  in the basis of cycle-type sums (see e.g. \cite{Las}).  The following result shows how to do this 
for the symmetric functions $g_j(\JJ, q, t)$.
\begin{theorem}
\label{gj_cycle_sums}
\be
g_j(\JJ, q, t)
 =\sum_{e=0}^j t^e  \sum_{k,  \, \tilde{k}=0}^{e, \,  j-e} {\hskip-30 pt}
\sideset{}{'}\sum_{\substack{\{\mu^{(u)}, \nu^{(v)} \}\\
1\le u\le k, 1\le v \le {\tilde{k}} \\
 |\mu^{(u)}|=| \nu^{(v})|=n \\
  \sum_{u=1}^k \ell^*(\mu^{(u})=e  \\ 
  \sum_{u=1}^k \ell^*(\mu^{(u)})+ \sum_{v=1}^{\tilde{k}} \ell^*(\nu^{(v)} )=j 
     }}
{\hskip-40 pt} W_q(\mu^{(1)}, \dots, \mu^{(k)}; \nu^{(1)}, \dots, \nu^{({\tilde{k}})})
 \prod_{u=1}^k C_{\mu^{(u)}} \prod_{v=1}^{\tilde{k}} C_{\nu^{(v)}}
  \label{gj_cycle_expansion}
\ee
where, by (\ref{combin_hurwitz}),
\be
\prod_{u=1}^k C_{\mu^{(u)}} \prod_{v=1}^{\tilde{k}} C_{\nu^{(v)}} =\sum_{\nu, \, |\nu|=n}
H(\mu^{(1)}, \dots , \mu^{(k)}, \nu^{(1)}, \dots , \nu^{({\tilde{k}})}, \nu) \, z_\nu C_\nu
\ee
\end{theorem}
\br
Note that the sums appearing in (\ref{gj_cycle_expansion}) are all finite because of the fact that
the partitions corresponding to the identity class of $S_n$  are excluded
and the constraints
\be
 \sum_{u=1}^k \ell^*(\mu^{(u})=e, \quad  
  \sum_{u=1}^k \ell^*(\mu^{(u)})+ \sum_{v=1}^{\tilde{k}} \ell^*(\nu^{(v)} )=j 
\ee
imply the number of partitions $k + \tilde{k}$ is finite.
\er
\begin{proof}
We start with the expansion
\be
\prod_{a=1}^n \prod_{k=0}^\infty {1- t z \JJ_a q^k \over 1 - z \JJ_aq^k}
=\sum_{j=0}^\infty g_j(\JJ, q, t) z^j
\label{Pi_z_JJ}
\ee
Applying the LHS  of (\ref{Pi_z_JJ}) to  $F_\lambda$ and using (\ref{content_product_eigenvalue}) gives
\bea
\prod_{a=1}^n \prod_{k=0}^\infty {1- t z \JJ_a q^k \over 1 - z \JJ_a q^k} F_\lambda
&\& = \prod_{(i,j)\in \lambda} \prod_{k=0}^\infty {1 - tz (j-i) q^k \over 1 - z (j-i) q^k} F_\lambda  = \prod_{k=0}^\infty {(-{1\over tz q^k})_\lambda \over (-{1\over z q^k})_\lambda}\, F_\lambda \\
&\& {} \cr
&\&=
 \prod_{k=0}^\infty {1+ h_\lambda\sideset{}{'} \sum_{\mu, \, \abs{\mu} =\abs{\lambda}}
 \frac{\chi_\lambda(\mu)}{z_\mu} (-tzq^k)^{\ell^*(\mu)} 
 \over 1+ h_\lambda\sideset{}{'} \sum_{\nu, \, \abs{\nu}=\abs{\lambda}}\frac{\chi_\lambda(\nu)}{z_\nu} 
 (-z q^k)^{\ell^*(\nu)}  }\, F_\lambda,
 \label{Pi_z_eigenvalue}
\eea
where Lemma \ref{Poch_Frob} has been used in both the numerator and denominator of (\ref{Pi_z_eigenvalue}).
From the relation  (\ref{C_mu_F_lambda}) and the fact that $\{F_\lambda\}$ is a basis for the center 
$\Zb(\Cb[S_n])$, eq.~(\ref{Pi_z_eigenvalue}),  together with (\ref{Pi_z_JJ}), is equivalent to the identity
\be
\sum_{j=0}^\infty g_j(\JJ, q, t) z^j =  \prod_{k=0}^\infty {1 + \sum C_\mu (-tzq^k)^{\ell^*(\mu)} \over
1 + \sum C_\nu (-zq^k)^{\ell^*(\nu)}}.
\label{ratio_cycle_sum_id}
\ee
Expanding (\ref{ratio_cycle_sum_id}) as a power series in $z$ and $t$,  and summing the resulting 
geometric series expansions in $q$, as detailed in \cite{GH2}, to obtain (\ref{W_E_q}) and (\ref{W_H_q})
gives the result (\ref{gj_cycle_expansion}).
\end{proof}

   Now let $\{\{\mu^{i, u_i}\}_{u_i =1, \dots , k_i}, \{\nu^{i, v_i}\}_{v_ i= 1, \dots, \tilde{k}_i}, \mu, \nu \}_{i=1, \dots, l}$ 
denote the branching profiles of an $n$-sheeted covering of the Riemann sphere  
with two specified branch points of ramification profile types $(\mu, \nu)$, at $(0, \infty)$,  and the rest divided into  two classes I and II,
denoted  $\{\mu^{(i,u_i)}\}_{u_i =1, \dots , k_i}$ and  $\{\nu^{(i, v_ i)}\}_{v_ i= 1, \dots, \tilde{k}_i}$, respectively. These are further subdivided into $l $ species, or ``colours'', labelled by  $i=1, \dots  l$,  the elements within each colour group distinguished  by the  labels $(u_i =1, \dots , k_i)$  
  and $(v_ i =1, \dots , \tilde{k}_ i)$.  To such a grouping, we assign a partition $\lambda$ of length
\be
\ell(\lambda) = l
\ee
and weight
\be
d: = |\lambda| = \sum_{i =1}^l \left (\sum_{u_i =1}^{k_i}\ell^*(\mu^{(i, u_i)})
+ \sum_{v_i =1}^{\tilde{k}_i}\ell^*(\nu^{(i, v_i)}) \right) = \sum_{i=1}^l d_i, 
\ee
whose parts $(\lambda_1\ge \cdots \ge \lambda_l > 0)$ are equal  the total colengths
\be
d_i := \sum_{u_i =1}^{k_i} \ell^*(\mu^{(i, u_i)}) + \sum_{v_i=1}^{\tilde{k}_i}\ell^*(\nu^{(i, v_i)}), \quad i=1. \dots, l
\ee
in weakly decreasing order.
By the Riemann-Hurwitz formula, the genus $g$ of the covering curve is given by
\be
  2-2g = \ell(\mu) +\ell(\nu)  - d.
  \ee 
  
We now assign a weight $W_q (\{\mu^{(i,u_i)}, \nu^{( i, v_ i)}\}, {\bf c} ) $ 
to each such covering, consisting of the product of all the weights 
$W_{E(q)}(\{\mu^{(i, u_i)}\}_{u_i = 1, \dots, k_i})$, 
$W_{H(q)}(\{\nu^{( i, v_ i)}\}_{v_ i = 1, \dots, \, \tilde{k}_ i})$
for the subsets of different colour and class and  the weight $m_\lambda({\bf c})$
given by the monomial symmetric functions evaluated at the parameters ${\bf c}$
\bea
  W_q (\{\mu^{(i,u_i)}, \nu^{( i, v_ i)}\}, {\bf c} ) 
  &\&:= W_q (\{\mu^{(i,u_i)}, \nu^{( i, v_ i)}\} ) m_\lambda({\bf c})\\
  W_q (\{\mu^{(i,u_i)}, \nu^{( i, v_ i)}\} ) &\&:= 
  \prod_{i=1}^lW_{E(q)}(\{\mu^{(i,u_i)}\}_{u_i = 1, \dots, \, k_i})
   W_{H(q)}(\{\nu^{( i, v_ i)}\}_{ i = 1, \dots, \,  \tilde{k}_ i})
   \eea

Using these weights, for every pair $(d,e)$ of non-negative integers and  $(\mu, \nu)$
of partitions of $n$,  we define the  geometrical  quantum weighted Hurwitz numbers 
$H^{(d,e)}_{({\bf c}, q)}(\mu, \nu) $ as the sum
\be
H^{(d,e)}_{({\bf c}, q)}(\mu, \nu) := z_\nu \sum_{l=0}^{d} {\hskip -35 pt}
\sideset{}{'} \sum_{\substack{\{\mu^{(i, u_i)}, \nu^{(i, v_i)}\} , \ k_i\ge 1, \ \tilde{k}_i \ge 1\\
\sum_{i=1}^l \sum_{u_i =1}^{k_i}\ell^*(\mu^{(i,u_i)}) = e, \\  
\sum_{i=1}^l\left( \sum_{u_i =1}^{k_i}\ell^*(\mu^{(i, u_i)} )
+  \sum_{v_ i =1}^{\tilde{k}_ i}\ell^*(\nu^{( i, v_ i)})\right) =d}} {\hskip - 20 pt}
{\hskip-50 pt}W_q (\{\mu^{(i, u_i)}, \nu^{( i, v_ i})\}, {\bf c}) \ 
H(\{\mu^{(i, u_i)}\}_{\substack{u_i = 1, \dots, k_i \\ i =1, \dots , l}} ,
\{\nu^{( i, v_ i)}\}_{\substack{v_ i = 1, \dots, \tilde{k}_ i \\  i =1, \dots , l}},  \mu, \nu).
\label{Hde_c_q_mu_nu}
\ee

\begin{theorem}
\label{geometric_hurwitz}
The combinatorial Hurwitz numbers $F^d_{M(q,t,{\bf c})} (\mu, \nu) $  are polynomials in $t$  of degree $d$  whose 
coefficients are equal to  the  geometrical  quantum weighted Hurwitz numbers $H^{(d,e)}_{({\bf c}, q)}(\mu, \nu)$
\be
F^d_{M(q,t,{\bf c})} (\mu, \nu) = \sum_{e=0}^d H^{(d,e)}_{({\bf c}, q)}(\mu, \nu) t^e.
\label{Fd_G_Hde_G} 
\ee

Hence $\tau^{M(q,t,{\bf c}, z)} ({\bf t}, {\bf s})$,  when expanded in
the basis of products of power sum symmetric functions and power series in $z$ and $t$
is the generating function for the  $H^{(d,e)}_{({\bf c}, q)}(\mu, \nu)$'s:
\be
\tau^{M(q,t,{\bf c}, z)} ({\bf t}, {\bf s}) = \sum_{d=0}^\infty \sum_{e=0}^d z^d t^e H^{(d,e)}_{({\bf c}, q)}(\mu, \nu) p_\mu({\bf t}) p_\nu({\bf s}).
\label{tau_G_H}
\ee
\end{theorem}
\begin{proof}
Substitution of (\ref{gj_cycle_expansion}) into
\be
g_\lambda (\JJ, q, t) = \prod_{i=1}^{\ell(\lambda)}g_{\lambda_i}(\JJ, q, t)
\ee
gives
\be
g_\lambda (\JJ, q, t)  =\sum_{e=0}^{|\lambda|}  t^e
 {\hskip -30 pt}
\sideset{}{'}\sum_{\substack{\{\mu^{(i, u_i)}, \nu^{( i, v_ i)}\} 
\\ \sum_{i=1}^{\ell(\lambda)} \sum_{u_i =1}^{k_i}\ell^*(\mu^{(i, u_i)}) = e, \\ 
 \sum_{u_i =1}^{k_i}\ell^*(\mu^{(i, u_i)}) +  \sum_{v_i =1}^{\tilde{k}_i}\ell^*(\nu^{(i, v_i)}) = \lambda_{i}}} {\hskip - 40 pt}
W_q (\{\mu^{(i, u_i)}, \nu^{( i, v_ i)}\} )  \prod_{i=1}^{\ell(\lambda)} \left(\prod_{u_i=1}^{k_i}
C_{\mu^{(i, u_i)}}\prod_{v_ i=1}^{\tilde{k}_ i }C_{\nu^{( i, v_ i)}} \right)
\ee
Combining this with (\ref{Pihat_geometrical}) gives
\be
M_n(q,t, {\bf c}, z\JJ)  C_\mu = \sum_{d=0}^\infty \sum_{e=0}^d z^d t^e \sum_{\substack{\nu\\ |\nu|=|\mu| = n}} H^{(d,e)}_{({\bf c}, q)} (\mu, \nu)  C_\nu, 
\label{Pi_Cmu_H_Gd}
\ee
where $H^{(d,e)}_{({\bf c}, q)} (\mu, \nu)$  is defined by (\ref{Hde_c_q_mu_nu}).
Comparing with (\ref{Pi_Cmu_F_Gd}) gives the result (\ref{Fd_G_Hde_G}), and hence (\ref{tau_G_H}).
\end{proof}

\section{Specializations, limits and examples}
 \label{examples}

By making specific  choices for the parameters $\{(c_1, c_2,  \dots ), q, t\}$ defining the weight
generating function $M(q,t,  {\bf c}, z)$, specialized versions of the above
quantum weighted Hurwitz numbers result.  Taking the limits $(z,t) \ra (0, \infty)$,  with  $t z$ fixed gives
the quantum deformation of the path weighting by elementary symmetric functions
considered in \cite{GH2}. The limit $t\ra 0$ gives the dual case, weighted by the quantum deformation of the
 path weighting by complete  symmetric functions. Other specializations involving only particular
values for the pair $(q,t)$ or their limits reduce the Macdonald polynomials either to Schur polynomials ($q=t)$,
or Hall-Littlewood polynomials ($q=0$) or Jack polynomials ($q=t^\alpha, \ t\ra 1$). 
In this way  we can recover  all previously studied versions of weighted 
Hurwitz numbers, as well as  several new examples of interest.


\subsection{Classically weighted Hurwitz numbers $(q=t)$}

By setting $t=q$ in (\ref{Pihat}) we recover the case of Schur functions  and the general
classically weighted families of Hurwitz numbers studied in \cite{GH2}.


\subsection{The case $c_i= -\delta_{i,1}$ (quantum monotonic paths) }

 This gives the quantum deformation of the  classical case (corresponding to $q=0$) when the weight generating function
  is ${1+w \over 1-z}$, with $w = -tz$.  If $w=0$, the latter  becomes the signed counting problem
  for branched covers with fixed genus or, equivalently,  weakly monotonic paths in
  the Cayley graph generated by transpositions \cite{GH1, GH2}. When $z=0$ it  gives  the  Hurwitz numbers
  for Belyi curves (having three branch points, with two of them fixed)  of fixed genus or, equivalently, 
 strongly  monotonic paths in  the Cayley graph generated by transpositions \cite{GH1, GH2}.
  When $q\neq 0,\,  t=1$, this is the particular case of the multispecies quantum Hurwitz numbers 
  $F^d_{Q(q,q)}(\mu, \nu)= H^d_{Q(q,q)}(\mu, \nu)$ developed in detail in \cite{GH2},  when there 
  are only two species involved, one of the first class, the other, of second.

\subsection{Elementary quantum weighting  ($(z,t) \ra (0, \infty)$,   $-t z $ fixed  ($\ra z$) }
\label{E_c_q}

For this case, the weight generating function is 
\be
E(q, {\bf c}, z) := \prod_{k=0}^\infty \prod_{i=1}^\infty(1 +zq^k c_i) = \prod_{i=1}^\infty (-zc_i; q)_{\infty}
=: \sum_{j=0}^\infty e_j(q,{\bf c}) z^j, 
\ee
where $e_j(q, {\bf c})$ is the quantum deformation of the elementary symmetric function $e_j({\bf c})$
(the classical limit being $q \ra 0$). Setting $c_i = \delta_{i1}$ reproduces the generating function
functions for the special quantum weighted Hurwitz numbers  denoted 
$H^d_{E(q)}(\mu, \nu) = F^d_{E(q)}(\mu, \nu)$ that were studied in \cite{GH2}.

In the general case, the corresponding element of the center of the group algebra is:
\be
E_n(q, {\bf c}, z\JJ)  := \prod_{a=1}^n E(q, {\bf c}, z\JJ_a) 
 =\sum_{\lambda} z^{|\lambda|} e_\lambda(q, {\bf c})  m_\lambda(\JJ) 
 = \sum_{\lambda} z^{|\lambda|} m_\lambda(\JJ)  e_\lambda(q, {\bf c}) 
\ee
where 
\be
e_\lambda(q, {\bf c}) := \prod_{i=1}^{\ell(\lambda)} e_{\lambda_i} (q,{\bf c}).
\ee

Applying $E_n(q, {\bf c}, z\JJ)  \in\Zb(\Cb[S_n])$ to the orthogonal idempotents
$\{F_\lambda\}$ and  the cycle-type sums $\{C_\mu\}$, it follows that
the corresponding hypergeometric $2D$ Toda $\tau$-function is 
\bea
\tau^{E(q, {\bf c}, z)}({\bf t}, {\bf s}) &\&= \sum_\lambda r_\lambda^{E(q,{\bf c}, z)} s_\lambda({\bf t}) s_\lambda({\bf s}) \\
 &\&= \sum_{d=0}^\infty  z^d\sum_\lambda F^d_{E(q, {\bf c})}(\mu, \nu)  p_\mu({\bf t}) p_\nu({\bf s}),
\eea
where the content product coefficient $r_\lambda^{E(q,{\bf c}, z)}$ is
\be
r_\lambda^{E(q, {\bf c}, z)} := \prod_{(ij) \in \lambda} \prod_{k=0}^\infty   (-z(j-i)c_k; q)_\infty
\ee
and
\be
F^d_{E(q, {\bf c})}(\mu, \nu) := \sum_{|\lambda|=d}e_\lambda(q, {\bf c}) m^\lambda_{\mu \nu}
\label{F_dE_qc}
\ee
is  the weighted number of paths in the Cayley graph of $S_n$ generated by transpositions,
starting at  the conjugacy class $\cyc(\mu)$ and ending at $\cyc(\nu)$, with the
weight  $e_\lambda(q, {\bf c})$  for a path of signature $\lambda$. 
  
 Now consider $n$-fold branched coverings of $\Cb \Pb^1$ with a fixed pair of branch points   at $(0, \infty)$
   with ramification profiles $(\mu, \nu)$ and  a further $  \sum_{i=1}^l k_i $  branch points  $\{\mu^{(i,u_i)}\}_{u_i = 1, \dots, \, k_i}$ of $l$ 
   different species (or ``colours''), labelled by $i=1, \dots , l$, with non trivial ramification profiles.
   The weight  $W_{E^l(q)} (\{\mu^{(i,u_i)}\}_{\substack{u_i = 1, \dots, k_i \\ i=1, \dots , l}}, {\bf c} )$
for such a covering consists of the product of all the weights 
$W_{E(q)}(\{\mu^{(i, u_i)}\}_{u_i = 1, \dots, k_i})$, 
for the subsets of different colour  with the weight $m_\lambda({\bf c})$
given by the monomial symmetric functions evaluated at the parameters ${\bf c}$
\bea
  W_{E^l(q)} (\{\mu^{(i,u_i)}\}_{\substack{u_i = 1, \dots, k_i \\ i=1, \dots , l}},  {\bf c}  )
  &\&:= W_{E^l(q)} (\{\mu^{(i,u_i)}\}_{\substack{u_i = 1, \dots, k_i \\ i=1, \dots , l}}  ) \, m_\lambda({\bf c})\\
  W_{E^l(q)}(\{\mu^{(i,u_i)}\}_{\substack{u_i = 1, \dots, k_i \\ i=1, \dots , l}} ) &\&:= 
  \prod_{i=1}^lW_{E(q)}(\{\mu^{(i,u_i)}\}_{u_i = 1, \dots, \, k_i}).
   \eea
We then have
\be
 F^d_{E(q,{\bf c})}(\mu, \nu) =H^d_{E(q, {\bf c})} (\mu, \nu), 
\ee
where
\be
H^d_{E({\bf c}, q)}(\mu, \nu) := z_\nu \sum_{l=0}^{d} {\hskip -10pt}
\sideset{}{'} \sum_{\substack{\{\mu^{(i, u_i)}\} , \ k_i\ge 1, \ \\
\sum_{i=1}^l \sum_{u_i =1}^{k_i}\ell^*(\mu^{(i,u_i)}) = d 
}} 
{\hskip-20 pt}W_{E^l(q)} (\{\mu^{(i, u_i)}\}_{\substack{u_i = 1, \dots, k_i \\ i=1, \dots , l}}, {\bf c}) \ 
H(\{\mu^{(i, u_i)}\}_{\substack{u_i = 1, \dots, k_i \\ i =1, \dots , l}} ,
 \mu, \nu)
\label{Eq_d_c}
\ee
  is the geometrical elementary quantum weighted Hurwitz number.

\subsection{Complete quantum weighting  ($t=0$)}
\label{H_c_q}

This is the dual of the preceding case, with weight generating function
\be
H(q, {\bf c}, z) := \prod_{k=0}^\infty \prod_{i=1}^\infty(1 - zq^k c_i)^{-1} = \prod_{i=1}^\infty (zc_i; q)^{-1}_{\infty}
 =: \sum_{j=0}^\infty h_j(q,{\bf c}) z^j,
\ee
where $h_j(q, {\bf c})$ is the quantum deformation of the complete symmetric function $h_j({\bf c})$.  
Setting $c_i = \delta_{i1}$ reproduces the generating function
functions for the  quantum weighted Hurwitz numbers  $H^d_{H(q)} (\mu, \nu)= F^d_{H(q)}(\mu, \nu)$ studied in \cite{GH2}.
The corresponding element of the center of the group algebra in the general case is:
\bea
H_n(q, {\bf c}, z\JJ) := \prod_{a=1}^n H(q, {\bf c}, \JJ_a) \
 =\sum_{\lambda}z^{|\lambda|}  h_\lambda(q, {\bf c} )m_\lambda(\JJ)
  = \sum_{\lambda} z^{|\lambda|} m_\lambda({\bf c})  h_\lambda(q, \JJ) ,
\eea
where 
\be
h_\lambda(q, {\bf c}) := \prod_{i=1}^{\ell(\lambda)} h_{\lambda_i} (q, {\bf c}).
\ee

The hypergeometric $2D$ Toda $\tau$-function for this case  is 
\bea
\tau^{H(q,{\bf c}, z)}({\bf t}, {\bf s}) &\&= \sum_\lambda r_\lambda^{H(q, {\bf c}, z)} s_\lambda({\bf t}) s_\lambda({\bf s}) \\
 &\&= \sum_{d=0}^\infty  z^d\sum_\lambda F^d_{H(q,{\bf c})}(\mu, \nu)  p_\mu({\bf t}) p_\nu({\bf s}), 
\eea
where 
\be
r_\lambda^{H(q,{\bf c}, z)} := \prod_{(ij) \in \lambda} \prod_{k=0}^\infty   (z(j-i)c_k; q)^{-1}_\infty
\ee
and
\be
F^d_{H(q,{\bf c})}(\mu, \nu) := \sum_{|\lambda|=d}h_\lambda(q, {\bf c}) m^\lambda_{\mu \nu}
\label{F_dH_qc}
\ee
is  the weighted number of paths in the Cayley graph of $S_n$ generated by transpositions,
starting at  the conjugacy class $\cyc(\mu)$ and ending at $\cyc(\nu)$, with 
weight  $h_\lambda(q, {\bf c})$  for a path of signature $\lambda$.
 
Consider again  $n$-fold branched coverings of $\Cb \Pb^1$,   with a fixed pair of branch points   at $(0, \infty)$
   with ramification profiles $(\mu, \nu)$ and  a further 
   $  \sum_{i=1}^l \tilde{k}_i $  branch points  $\{\nu^{(i,v_i)}\}_{v_i = 1, \dots, \, \tilde{k}_i}$
   again,  of $l$  different species (or ``colours''), labelled by $i=1, \dots , l$, with nontrivial ramification profiles.
   Like in the preceding case, the weight  $W_{H^l(q)} (\{\nu^{(i,v_i)}\}_{\substack{v_i = 1, \dots, \tilde{k}_i \\ i=1, \dots , l}}, {\bf c} )$
for such a covering consists now of the product of all weights 
$W_{H(q)}(\{\nu^{(i, v_i)}\}_{v_i = 1, \dots, \tilde{k}_i})$, 
for the subsets of different colour  with the weight $m_\lambda({\bf c})$
again given by the monomial symmetric functions evaluated at the parameters ${\bf c}$
\bea
  W_{H^l(q)} (\{\nu^{(i,v_i)}\}_{\substack{v_i = 1, \dots, \tilde{k}_i \\ i=1, \dots , l}},  {\bf c}  )
  &\&:= W_{H^l(q)} (\{\nu^{(i,v_i)}\}_{\substack{v_i = 1, \dots, \tilde{k}_i \\ i=1, \dots , l}} ) \, m_\lambda({\bf c})\\
  W_{H^l(q)}(\{\nu^{(i,v_i)}\}_{\substack{v_i = 1, \dots, \tilde{k}_i \\ i=1, \dots , l}} ) &\&:= 
  \prod_{i=1}^lW_{H(q)}(\{\nu^{(i,v_i)}\}_{v_i = 1, \dots, \, \tilde{k}_i}).
   \eea
We again have the equality
\be
 F^d_{H(q,{\bf c})}(\mu, \nu) =H^d_{H(q, {\bf c})} (\mu, \nu), 
\ee
where
\be
H^d_{H({\bf c}, q)}(\mu, \nu) := z_\nu \sum_{l=0}^{d} {\hskip -10pt}
\sideset{}{'} \sum_{\substack{\{\nu^{(i, v_i)}\} , \ \tilde{k}_i\ge 1, \ \\
\sum_{i=1}^l \sum_{v_i =1}^{\tilde{k}_i}\ell^*(\nu^{(i,v_i)}) = d 
}} 
{\hskip-20 pt}W_{H^l(q)} (\{\nu^{(i, v_i)}\}_{\substack{v_i = 1, \dots, \tilde{k}_i \\ i=1, \dots , l}}, {\bf c}) \ 
H(\{\nu^{(i, v_i)}\}_{\substack{v_i = 1, \dots, \tilde{k}_i \\ i =1, \dots , l}} ,
 \mu, \nu)
\label{Eq_d_c}
\ee
  is the corresponding geometrically defined complete quantum weighted Hurwitz number.

\subsection{Hall-Littlewood polynomials ($q =0$) }

Setting $q=0$ in eq.~(\ref{Pi_xy_qt}), the generating function reduces to the one for Hall-Littlewood polynomials
 \cite[Sec. III.2]{Mac} $P_\lambda({\bf x}, t)$, which satisfy the orthogonality relations
\be
(P_\lambda, P_\mu)_t= \delta_{\lambda \mu} (b_\lambda(t))^{-1},  \quad  b_\lambda (t) := \prod_{i\ge 1} \prod_{k=1}^{m_i(\lambda)} (1 - t^k)
\ee
 with respect to the scalar product $(\  , \ )_t$ defined by
 \be
 (p_\lambda, p_\mu)_t =  \delta_{\lambda \mu} z_\lambda n_\lambda (t), \quad n_\lambda:= \prod_{i=1}^{\ell(\lambda)}  {1\over 1 - t^{\lambda_i}}.
 \ee
 
 Following \cite{Mac}, we define
 \be
 q_\lambda({\bf x}, t) := b_\lambda(t) \prod_{i=1}^{\ell(\lambda)} P_j({\bf x}, t)
 \label{q_lambda}
 \ee
and obtain the following expansion
 \be
L(t, {\bf x}, {\bf y}):= \prod_{i, j}^\infty {1 - t x_i y_j   \over 1- x_i y_j}
=  \sum_{\lambda}^\infty  q_\lambda({\bf x}, t) m_\lambda({\bf y}) 
= \sum_{\lambda}^\infty   q_\lambda({\bf y}, t) m_\lambda({\bf x}).
\label{Lt_yx_t}
\ee
Substituting ${\bf c}=(c_1, c_2, \dots )$  for ${\bf x}$,  and $(\JJ_1, \dots , \JJ_n)$ for ${\bf y}$, we have 
\be
L(t, {\bf c}, z\JJ):=  \prod_{i=1}^\infty \prod_{a=1}^n  {1 - t c_i z \JJ_a   \over 1- c_i z \JJ_a}
=  \sum_{\lambda}^\infty z^{|\lambda|}q_\lambda({\bf c}, t) m_\lambda({\bf \JJ})
 = \sum_{\lambda}^\infty z^{|\lambda|} q_\lambda({\bf \JJ}, t) m_\lambda({\bf c}).
\label{HL_generating_element}
\ee

Applying  $L(t, {\bf c}, z\JJ),\in \Zb(\Cb[S_n])$ to the orthogonal idempotents
$\{F_\lambda\}$ and the cycle-type sums $\{C_\mu\}$ as above,
the corresponding hypergeometric $2D$ Toda $\tau$-function becomes
\bea
\tau^{L(t,{\bf c}, z)}({\bf t}, {\bf s}) &\&= \sum_\lambda r_\lambda^{L(t,{\bf c}, z)} s_\lambda({\bf t}) s_\lambda({\bf s}) \\
 &\&= \sum_{d=0}^\infty  z^d\sum_\lambda F^d_{L(t,{\bf c})}(\mu, \nu)  p_\mu({\bf t}) p_\nu({\bf s})\\
\eea
where 
\be
r_\lambda^{L(t,{\bf c}, z)} := \prod_{(ij) \in \lambda} \prod_{k=1}^\infty  { 1 - t z (j-i) c_k \over 1 - z (j-i) c_k} 
= \prod_{k=1}^\infty (-t)^{|\lambda|} {(- 1/(tz c_k))_\lambda \over (- 1/(z c_k))_\lambda }
\ee
and
\be
F^d_{L(t,{\bf c})}(\mu, \nu) := \sum_{|\lambda|=d}q_\lambda({\bf c}, t) m^\lambda_{\mu \nu}
\ee
is  again the weighted number of paths in the Cayley graph of $S_n$ generated by transpositions,
with weight  $q_\lambda({\bf c},t)$  for a path of signature $\lambda$.

We also have
\be
 F^d_{L(t,{\bf c})}(\mu, \nu) = \sum_{e=0}^d H^{(d,e)}_{L({\bf c})} (\mu, \nu) t^e
\ee
where
\be
H^{(d,e)}_{L({\bf c})}(\mu, \nu) := z_\nu \sum_{l=0}^{d} {\hskip -35 pt}
\sideset{}{'} \sum_{\substack{\{\mu^{(i, u_i)}, \nu^{(i, v_i)}\} , \ k_i\ge 1, \ \tilde{k}_i \ge 1\\
\sum_{i=1}^l \sum_{u_i =1}^{k_i}\ell^*(\mu^{(i,u_i)}) = e, \\  
\sum_{i=1}^l\left( \sum_{u_i =1}^{k_i}\ell^*(\mu^{(i, u_i)} )
+  \sum_{v_ i =1}^{\tilde{k}_ i}\ell^*(\nu^{( i, v_ i)})\right) =d}} {\hskip - 20 pt}
{\hskip-50 pt} (-1)^{K+d-e}
H(\{\mu^{(i, u_i)}\}_{\substack{u_i = 1, \dots, k_i \\ i =1, \dots , l}} ,
\{\nu^{( i, v_ i)}\}_{\substack{v_ i = 1, \dots, \tilde{k}_ i \\  i =1, \dots , l}},  \mu, \nu),
\label{H_de_c}
\ee
with 
\be
K:= \sum_{i=1}^l (k_i +\tilde{k}_i)
\ee
the total number of branch points. $H^{(d,e)}_{\bf c}(\mu, \nu)$  is the weighted generalization of the multispecies hybrid signed Hurwitz numbers studied in \cite{HO2}.
As in the general Macdonald  case,  $H^{(d,e)}_{L({\bf c})}(\mu, \nu)$ is the weighted number of $n$-fold branched coverings of $\Cb \Pb^1$ with a fixed pair of branch points  with ramification profiles $(\mu, \nu)$,  and $K$ additional branch points   divided into  two classes I and II, denoted  $\{\mu^{(i,u_i)}\}_{u_i =1, \dots , k_i}$ and  $\{\nu^{(i, v_ i)}\}_{v_ i= 1, \dots, \tilde{k}_i}$,  respectively, which  are further subdivided into $l $ species, or ``colours'', labelled by  $i=1, \dots  l$,   the elements within each colour group distinguished  by the  labels $(u_i =1, \dots , k_i)$  
  and $(v_ i =1, \dots , \tilde{k}_ i)$.  To such a grouping, we again assign a partition $\lambda$ of length
\be
\ell(\lambda) = l
\ee
and weight
\be
d: = |\lambda| = \sum_{i =1}^l \left (\sum_{u_i =1}^{k_i}\ell^*(\mu^{(i, u_i)})
+ \sum_{v_i =1}^{\tilde{k}_i}\ell^*(\nu^{(i, v_i)}) \right) = \sum_{i=1}^l d_i, 
\ee
whose parts $(\lambda_1\ge \cdots \ge \lambda_l > 0)$ are equal  the total colengths
\be
d_i := \sum_{u_i =1}^{k_i} \ell^*(\mu^{(i, u_i)}) + \sum_{v_i=1}^{\tilde{k}_i}\ell^*(\nu^{(i, v_i)}), \quad i=1. \dots, l
\ee
in weakly decreasing order.

\subsection{Jack polynomials ($ q= t^\alpha, \, t \ra 1$)}
\label{jack}

Setting $q=t^\alpha$ and taking the limit $q \ra 1$, we obtain the Jack polynomials \cite[Sec. VI.10]{Mac}  $P^{(\alpha)}_\lambda$ as the limiting 
case of the MacDonald polynomials. 
These satisfy the orthogonality relations
\be
\langle P^{\alpha}_\lambda, P^{\alpha}_\mu\rangle_\alpha= \delta_{\lambda \mu} z_\lambda (b_\lambda^{(\alpha)})^{-1},  
\quad  b_\lambda^{(\alpha)}:= \prod_{i=1}^{\ell(\lambda)} \prod_{j=1}^{\lambda_i} { \alpha(\lambda_i -j) +\lambda'_j -i +1
\over  \alpha(\lambda_i -j) +\lambda_j -i +\alpha}
\ee
 with respect to the scalar product $\langle \  , \ \rangle_\alpha$ defined by
 \be
 \langle  p_\lambda, p_\mu\rangle_\alpha =  \delta_{\lambda \mu} z_\lambda \alpha^{\ell(\lambda)}. 
 \ee
 
 This corresponds to the family of weight generating functions
\be
J(\alpha, {\bf c}, z) :=\prod_{k=1}^\infty (1 - z c_i)^{-1/\alpha}
\ee
and the corresponding family of central elements
\bea
J(\alpha, {\bf c}, z \JJ) := \prod_{i=1}^\infty \prod_{a=1}^n(1 - z c_i \JJ_a)^{-1/\alpha} 
 = \sum_{\lambda} z^{|\lambda|} g^{\alpha}_\lambda(\JJ) m_\lambda({\bf c})
   = \sum_{\lambda} z^{|\lambda|} g^{\alpha}_\lambda({\bf c}) m_\lambda({\JJ}) ,
\eea
where the symmetric functions $g^{\alpha}_\lambda({\bf x})$ are the analogs of  the $e_\lambda({\bf x})$ or $h_\lambda ({\bf x})$ bases formed from products of  elementary or  complete symmetric functions in the case of Schur functions  ($\alpha=1$),
\be
g^{\alpha}_\lambda({\bf x}) = \alpha^{\ell(\lambda)}\prod_{i=1}^{\ell(\lambda} P^{(\alpha)}_{(\lambda_i)}({\bf x}),
\ee

The content product coefficients entering in the double Schur function 
expansion of the associated hypergeometric $2D$ Toda $\tau$-functions
\be
\tau^{J(\alpha, {\bf c}, z)}({\bf t}, {\bf s}) = \sum_{\lambda} r_\lambda^{J(\alpha, {\bf c}, z)} s_\lambda({\bf t}) s_\lambda({\bf s})
\ee
in this case are
\be
r_\lambda^{J(\alpha, {\bf c}, z)}  = \prod_{(ij)\in \lambda} \prod_{k=0}^\infty (1- z (j-i) c_k)^{-1/\alpha}
= \prod_{k=1}^\infty (1-zc_k)^{|\lambda|\over \alpha} (-1/z c_k)_\lambda^{-1/\alpha}.
\ee
The expansion in the basis of products of power sum  symmetric functions is therefore
\be
\tau^{J(\alpha, {\bf c}, z)} ({\bf t}, {\bf s})
= \sum_{d=0}^\infty \sum_{\substack{\mu, \nu \\ \abs{\mu} =
 \abs{\nu}=n}} z^d F^d_{J(\alpha, {\bf c})}(\mu, \nu) p_\mu({\bf t}) p_\nu({\bf s})
\label{tau_GJ_F}
\ee
where
\be
F^d_{J(\alpha, {\bf c})} (\mu, \nu) = \sum_{\lambda} g^{\alpha}_\lambda({\bf c})m^\lambda_{\mu \nu}
\ee
is the combinatorial Hurwitz number giving the weighted number of  $d$-step paths of signature $\lambda$ in the Cayley
graph of $S_n$ , starting in the conjugacy class $\cyc(\mu)$ and ending in $\cyc(\nu)$, with weight $g^{\alpha)}_\lambda({\bf c})$.

We again have
\be
F^d_{J(\alpha, {\bf c})} (\mu, \nu)  = H^d_{J(\alpha, {\bf c})} (\mu, \nu) ,
\ee
where the weighted  geometrical Hurwitz number is
\be
H^d_{J(\alpha, {\bf c})} (\mu, \nu)  :=  \sum_{k=0}^\infty \left({-{1\over \alpha} \atop k}\right)
  \sum_{\substack{\mu^{(1)}, \dots , \mu^{((k)} \\ |\mu^{(i)}|=n  \\ \sum_{i=1}^k \ell^*(\mu^{(i)})=d }}
 m_\lambda({\bf c}) H(\mu^{(1)}, \dots , \mu^{(k)}, \mu, \nu),
\ee
with the sum over partitions $\lambda$ of length $k$,  and weight $d$ whose parts are $\{\ell^*(\mu^{(1)}), \dots, \ell^*(\mu^{(k)})\}$.

\bigskip

 \bigskip 
\noindent 
  {\small {\it Acknowledgements.}  This work  extends the  approach to the construction
  of parametric families of $\tau$-functions as generating functions for weighted Hurwitz numbers 
initiated jointly with  M. Guay-Paquet and extended to the multispecies case with
A. Yu. Orlov. The author is indebted to both for helpful discussions that helped clarify many of
the ideas  and methods underlying  this approach. }

\bigskip

\newcommand{\arxiv}[1]{\href{http://arxiv.org/abs/#1}{arXiv:{#1}}}

\end{document}